\documentclass[aps,pra,superscriptaddress,twocolumn,longbibliography,floatfix]{revtex4-2}

\usepackage{units}
\usepackage{amsmath}
\usepackage{amsthm}
\usepackage{amssymb}
\usepackage{graphicx}
\usepackage{color}
\usepackage{xcolor}
\usepackage{bbold}
\usepackage{titlesec} 

\usepackage{subfigure}

\definecolor{myurlcolor}{rgb}{0,0,0.7}
\definecolor{myrefcolor}{rgb}{0.1,0,0.9}

\usepackage[
	breaklinks,
	pdftex,
	colorlinks=true, 
	linkcolor=myrefcolor,
	citecolor=myurlcolor,
	urlcolor=myurlcolor
]{hyperref}

\newcommand\norm[1]{\lVert#1\rVert}

\usepackage[linesnumbered, ruled,vlined]{algorithm2e}

\renewcommand{\eqref}[1]{Eq.~(\ref{#1})} %

\newtheorem{thm}{\protect\theoremname}

\ifx\proof\undefined
\newenvironment{proof}[1][\protect\proofname]{\par
	\normalfont\topsep6\p@\@plus6\p@\relax
	\trivlist
	\itemindent\parindent
	\item[\hskip\labelsep\scshape #1]\ignorespaces
}{%
	\endtrivlist\@endpefalse
}
\providecommand{\proofname}{Proof}
\fi

\newtheorem*{result*}{Main Result}

\usepackage{bbm}

\usepackage{bm}

\usepackage{mathrsfs}

\makeatother

\providecommand{\factname}{Fact}
\providecommand{\theoremname}{Theorem}
\providecommand{\claimname}{Claim}
\providecommand{\lemmaname}{Lemma}
\providecommand{\definitionname}{Definition}
\providecommand{\corollaryname}{Corollary}

\definecolor{KB}{rgb}{0.4,0.3,0.9}

\definecolor{RAc}{rgb}{0.9,0.3,0.2}

\newcommand{\sectionMain}[1]{
\let\oldaddcontentsline\addcontentsline%
\renewcommand{\addcontentsline}[3]{}%
\section{#1}
\let\addcontentsline\oldaddcontentsline
}

\begin{document}

%
\title{Self-testing of genuine multipartite non-local and non-maximally entangled states}


\author{Ranendu Adhikary}
 \email{ronjumath@gmail.com}
\affiliation{Electronics and Communication Sciences Unit, Indian Statistical Institute, 203 B.T. Road, Kolkata 700108, India.}

\begin{abstract}
Self-testing enables the characterization of quantum systems with minimal assumptions on their internal working as such it represents the strongest form of certification for quantum systems. In the existing self-testing literature, self-testing states that are not maximally entangled, but exhibit genuine multipartite nonlocality, have remained an open problem. This is particularly important because, for many-body systems, genuine multipartite nonlocality has been recognized as the strongest form of multipartite quantum correlation. In this work, we present a Cabello-like paradox for scenarios involving an arbitrary number of parties. This paradox is a tool for detecting genuine multipartite nonlocality, allowing for the specific identification and self-testing of states that defy the paradox's limits the most, which turn out to be non-maximally multipartite entangled states. While recent results  [\textit{Šupić et al., Nature Physics, 2023}] suggest network self-testing as a means to self-test all quantum states, here we operate within the standard self-testing framework to self-test genuine multipartite non-local and non-maximally entangled states.

\end{abstract}
\maketitle

\paragraph{Introduction.---}

Since the 1920s, quantum mechanics has not only revolutionized our understanding of fundamental physics but has also driven technological progress. This field has been instrumental in the development of everyday technologies, such as lasers, transistors, and the Global Positioning System (GPS). In recent years, advancements in quantum theory have led to the emergence of innovative applications like quantum computing, quantum networks, and quantum-encrypted communication. At the forefront of these advancements is the phenomenon of quantum entanglement~\cite{nielsen2010quantum,bengtsson2017geometry}, a unique aspect of quantum theory that was recognized by the 2022 Nobel Prize for its critical role in quantum information theory and the foundations of quantum mechanics. The exploration of quantum entanglement's complex theory remains a dynamic and rich area of research, full of unanswered questions that, once solved, could pave the way for new technological applications~\cite{horodecki2009quantum,berta2023tangled,erhard2020advances}.

As the scale of quantum devices increases, certifying their internal workings becomes increasingly complex due to the vastness of the associated Hilbert space~\cite{eisert2020quantum}. To fully describe an \( n \)-qubit pure state, one needs \( 2^n - 1 \) complex parameters. Quantum state tomography, a common approach for system certification, involves obtaining the description of a quantum state by performing measurements on an ensemble of identical quantum states. However, estimating an unknown quantum state within a dimension \( d \) to an error \( \varepsilon \) (in \( l_1 \) norm) requires \( \Theta\left(\frac{d^2}{\varepsilon^{2}}\right) \) copies of that state~\cite{o2016efficient}. This requirement makes tomography impractical for large systems. Additionally, tomography is device-dependent, assuming characterized measurements, which is not always feasible in scenarios involving third-party devices.

The necessity to overcome these challenges has led to the development of \textit{self-testing protocols}~\cite{MY2004}.  Crucially, these protocols are device-independent, meaning they do not require any prior knowledge or characterization of the quantum devices being tested. Moreover, the number of measurements scales linearly with number of qubits. This makes them highly valuable for scalable and assumption-minimal certification of quantum systems~\cite{mckague2012robust,bamps2015sum,wang2016all,kaniewski2017self,vsupic2016self,bowles2018self,kaniewski2019maximal,sarkar2021self,rai2021device,valcarce2022self,rai2022self,mckague2012robust,yang2013robust,bharti2019robust,bharti2019local,shrotriya2021robust,shrotriya2022certifying,bharti2022graph,xu2023state,bharti2022graph,yang2014robust,kaniewski2016analytic,baccari2020scalable,bancal2015physical,coladangelo2017all,vsupic2023quantum,panwar2023elegant}. In their study, Mayers and Yao demonstrated that certain statistical correlations exist which necessitate that black boxes replicating these correlations must inherently produce the maximally entangled state \( |\phi_{+}\rangle = \frac{1}{\sqrt{2}} (|00\rangle + |11\rangle) \). Consequently, these Bell correlations serve as a self-test for the state \( |\phi_{+}\rangle \). The entanglement in the quantum state allows for correlations that allow the underlying entangled state to be self-tested.

A natural question arises: which entangled states can be self-tested?  In a breakthrough paper, Coladangelo et al. demonstrated that all pure bipartite states are amenable to self-testing~\cite{coladangelo2017all}. However, in the broader context of multipartite states, our understanding remains limited. This lack of comprehensive knowledge is not unexpected, given that entanglement presents more complexity and challenges in characterization when it involves multiple parties. In particular, self-testing for entangled states in the multipartite setting that admit \textit{genuine multipartite nonlocality}~\cite{svetlichny1987distinguishing} has remained poorly understood. Genuine multipartite nonlocality is a stronger form of nonlocality and understanding entangled states that manifest them can deepen our understanding of quantum theory.  For multipartite systems, genuine multipartite nonlocality has been recognized as the strongest form of multipartite quantum correlation~\cite{brunner2014bell,scarani2019bell}. Previous works focused on self-testing for maximally entangled states \cite{panwar2023elegant} that admit genuine multipartite nonlocality and it is unclear whether the non-maximally entangled states that admit genuine multipartite nonlocality can be self-tested.

In this paper, we answer the above question in the affirmative. We present a Cabello-like paradox for a scenario involving an arbitrary number of participants, each having access to two measurements that can yield any arbitrary number of outcomes. Our paradox helps witness genuine multipartite nonlocality. The state maximally violating the constraints of our paradox can be self-tested. Regarding robust self-testing, we provide a partial robustness result in the three-party scenario where we consider the non-ideal state to be a three-qubit state. Moreover, it is important to note that the aforementioned state is non-maximally entangled. We highlight that unlike Ref~\cite{vsupic2023quantum}, our work does not require quantum networks. We operate in the standard Bell-test based self-testing framework and do not need any additional assumptions unlike Ref~\cite{vsupic2023quantum}.

\paragraph{Genuine tripartite non-local correlation.---}\label{sec1}

Consider three distant observers denoted as $\mathcal{A}_{1}$, $\mathcal{A}_{2}$ and $\mathcal{A}_{3}$ engaged in rounds of measurements on tripartite quantum systems. In each round, every party selects a local measurement denoted as $\hat{x}_i$, and subsequently records the resulting outcome as $x_i$. The resultant joint conditional probability distribution, denoted as $P_{A_{1}A_{2}A_{3}}(x_1 x_2 x_3 | \hat{x}_1 \hat{x}_2 \hat{x}_3)$, is classified as local if it can be factorized. This factorization holds under the additional condition of possessing knowledge about a potentially concealed shared classical cause $\lambda$:

{\small \begin{equation}\label{eq:local_equation}
\begin{split}
&P_{\mathcal{L}}(x_1 x_2 x_3 | \hat{x}_1 \hat{x}_2 \hat{x}_3)  \\
& = \sum_{\lambda} \nu(\lambda) P_{A_{1}}(x_1 | \hat{x}_1, \lambda)P_{A_{2}}(x_2 | \hat{x}_2, \lambda)P_{A_{3}}(x_3 | \hat{x}_3, \lambda)  
\end{split}  
\end{equation}}

The discrete random variable representing the common cause is denoted as $\lambda$, with a distribution $\nu(\lambda) \geq 0$ and the normalization condition $\sum_{\lambda} \nu(\lambda) = 1$. The conditional probability distribution $P_{A_{i}}(x_i | \hat{x}_i, \lambda)$ characterizes the behavior of party $\mathcal{A}_{i}$. A distribution $P(x_1 x_2 x_3 | \hat{x}_1 \hat{x}_2 \hat{x}_3)$ is considered non-local if it cannot be decomposed as described in equation (\ref{eq:local_equation}). Given the measurement setups, we work under the assumption of the NS principle \cite{ghirardi1980general}. This implies that party $\mathcal{A}_{1}$ is unable to communicate with other parties through the choice of measurement. Specifically, the equality $P_{A_{2}A_{3}}(x_2 x_3 | \hat{x}_2 \hat{x}_3) = P_{A_{2}A_{3}}(x_2 x_3 | \hat{x}_1 \hat{x}_2 \hat{x}_3) = \sum_{x_1}P(x_1 x_2 x_3 | \hat{x}_1 \hat{x}_2 \hat{x}_3), \forall x_1$ holds, and the analogous statements apply to parties $\mathcal{A}_{2}$ and $\mathcal{A}_{3}$.

Suppose we relax the locality assumption such that any pair of parties can group and share non-local resources. This type of hybrid local and non-local model leads to joint conditional probability distributions,

{\small \begin{equation}\label{eq:bi_equation}
\begin{split}
&P_{2\ versus\ 1}(x_1 x_2 x_3 | \hat{x}_1 \hat{x}_2 \hat{x}_3)  \\
& = \sum_{\lambda_1} \nu_1(\lambda_1) P_{A_{1}A_{2}}(x_1 x_2|\hat{x}_1 \hat{x}_2, \lambda_1)P_{A_{3}}(x_3 | \hat{x}_3, \lambda_1)  \\ 
&+ \sum_{\lambda_2} \nu_2(\lambda_2) P_{A_{2}A_{3}}(x_2 x_3|\hat{x}_2 \hat{x}_3, \lambda_2)P_{A_{1}}(x_1 | \hat{x}_1, \lambda_2)  \\
&+ \sum_{\lambda_3} \nu_3(\lambda_3) P_{A_{1}A_{3}}(x_1 x_3|\hat{x}_1 \hat{x}_3, \lambda_3)P_{A_{2}}(x_2 | \hat{x}_2, \lambda_3)   
\end{split}  
\end{equation}}

with \(\nu_i(\lambda_i) \geq 0\) and \(\sum_{i, \lambda_i} \nu_i(\lambda_i) = 1\).
Distributions denoted as $P(x_1 x_2 x_3 | \hat{x}_1 \hat{x}_2 \hat{x}_3)$ that cannot be expressed in the form described in equation (\ref{eq:bi_equation}) are termed as genuine tripartite non-local distributions. According to findings in \cite{gallego2012operational}, the original concept of multipartite non-locality introduced by Svetlichny \cite{svetlichny1987distinguishing} encounters practical challenges when operationalized. To address these challenges, it is assumed that the NS principle \cite{popescu1994quantum} also holds at the level of distributions $P_{A_{i}A_{j}}(x_i x_j| \hat{x}_i \hat{x}_j, \lambda)$, implying that the conditional probabilities $P(x_i|\hat{x}_i, \lambda) = P(x_i| \hat{x}_i \hat{x}_j, \lambda) = \sum_{x_j}P(x_i x_j| \hat{x}_i \hat{x}_j, \lambda), \forall \hat{x}_j$, are well defined for all $\lambda$.

Let us consider a scenario where all three parties Alice, Bob, and Charlie have two possible measurement choices $\hat{x}_1 \in \{A_0, A_1\}$, $\hat{x}_2 \in \{B_0, B_1\}$ and $\hat{x}_3 \in \{C_0, C_1\}$ respectively. Each measurement has two possible outcomes, say $\{+,-\}$. We define a tripartite correlation $P(x_1, x_2, x_3 | \hat{x}_1, \hat{x}_2, \hat{x}_3)$ that satisfies the following equality constraint,

{\small \begin{equation}\label{eq:cab3}
\begin{split}
& P_{AB}(++ | A_1B_0) = 0, \\
& P_{BC}(++ | B_1C_0) = 0, \\
& P_{AC}(++ | A_0C_1) = 0. \\
\end{split}  
\end{equation}}

Let us define $p_3 \equiv P(+++ | A_0B_0C_0)$ and $q_3 \equiv P(--- | A_1B_1C_1)$. We define their difference to be $\mathcal{S}_3:=p_3-q_3$. We will first show that no local theory can satisfy $\mathcal{S}_3 > 0$. Further, we will see that it is not just any non-local correlation; rather, it is a genuine non-local correlation.

\begin{thm}\label{theorem1}
If a non-local correlation satisfies the conditions in (\ref{eq:cab3}) with $\mathcal{S}_3 > 0$, then it is a genuine tripartite non-local correlation.    
\end{thm}
\begin{proof}
    We will first show that no local theory can satisfy conditions (\ref{eq:cab3}) with $\mathcal{S}_3 > 0$. Let us define a local hidden variable $\lambda$ that captures the full description of the system. This $\lambda$ takes values from a set $\Omega$. Moreover, we assert that the comprehensive depiction of the joint system is denoted by $\nu(\lambda)$. In a local description, the joint probability can be expressed as follows: $P(x_1, x_2, x_3 | \hat{x}_1, \hat{x}_2, \hat{x}_3) = \sum_{\lambda \in \Omega}\nu(\lambda) \prod_{j=1}^{3} p(x_j | \hat{x}_j, \lambda)$, where $\nu(\lambda) \geq 0$, $\sum_{\lambda}\nu(\lambda) = 1$. Let's assume that $\mathcal{S}_3 > 0$, which indicates that $p_3$ is greater than $0$. This implies the existence of a nonempty subset $\Lambda \subseteq \Omega$, such that for all $\lambda \in \Lambda$, $P(+|A_0,\lambda)=P(+|B_0,\lambda)=P(+|C_0,\lambda)=1$. By equation (\ref{eq:cab3}) we have $P(-|A_1, \lambda) =P(-|B_1, \lambda) =P(-|C_1, \lambda) = 1$. Therefore, we can deduce that $P(-, -, -|A_1, B_1, C_1, \Lambda) = P(+, +, +|A_0, B_0, C_0)$. However, since $\Lambda \subseteq \Omega$, it follows that $P(-, -, -|A_1, B_1, C_1, \Lambda) \leq P(-, -, -|A_1, B_1, C_1)$. In other words, $\mathcal{S}_3 \leq 0$, which contradicts the initial assumption of $\mathcal{S}_3 > 0$. Therefore, we can conclude that if any correlation satisfies (\ref{eq:cab3}) and gives $\mathcal{S}_3 > 0$, then it must be non-local.

Now we will show that the non-locality we have is a genuine one. If this correlation were not genuine, it could be broken down according to the expression provided in equation (\ref{eq:bi_equation}). Given that $\mathcal{S}_3 > 0$, at least one component within the decomposition must yield a non-zero value for $P(+++ | A_0B_0C_0)$. Let's consider this specific component to be the initial one: $P_{AB}(x_1 x_2|\hat{x}_1\hat{x}_2, \lambda_1)P_{C}(x_3 | \hat{x}_3, \lambda_1)$. Consequently, $P_{AB}(++|A_0B_0, \lambda_1)P_{C}(+|C_0, \lambda_1) \neq 0$. The condition $P_{BC}(++ | B_1C_0) = 0$ in equation (\ref{eq:cab3}) leads to the conclusion that $P_{B}(+|B_1, \lambda_1)P_{C}(+|C_0, \lambda_1)$ must be zero. This, in turn, implies that $P_{B}(+|B_1, \lambda_1) = 0$, given that $P_{C}(+|C_0, \lambda_1)$ cannot be zero. Therefore, we deduce that $P_{B}(-|B_1, \lambda_1) = 1$. Now as the NS principle is assumed at the level of distribution $P_{AB}(x_1 x_2| \hat{x}_1 \hat{x}_2, \lambda_1)$, a deterministic marginal of $P_{AB}(x_1 x_2| \hat{x}_1 \hat{x}_2, \lambda_1)$ indicates that $P_{AB}(x_1 x_2|\hat{x}_1\hat{x}_2, \lambda_1) = P_{A}(x_1 |\hat{x}_1, \lambda_1)P_{B}(x_2|\hat{x}_2, \lambda_1)$. Similarly, all the other components within the decomposition become fully product-based. As a result, the distributions $P_{ABC}(x_1 x_2 x_3 | \hat{x}_1 \hat{x}_2 \hat{x}_3)$ transform into fully local distribution, which contradicts the initial premise considering $\mathcal{S}_3 > 0$.  
\end{proof}

\paragraph{Genuine $N$-partite non-local correlation.---}\label{sec2}

Let's consider a scenario where $N$ subsystems are distributed among $N$ separate parties $\mathcal{A}_1$, $\mathcal{A}_2$, $\ldots$, and $\mathcal{A}_N$. Each party, $\mathcal{A}_i$, can measure one of two observables: either $\hat{u}_i$ or $\hat{v}_i$, on their local subsystem. The measurement outcomes, denoted as $x_i$, can take values from 1 to $d_i$, where $d_i$ represents the dimension of the Hilbert space associated with the $i$-th subsystem $\mathcal{A}_i$. We are interested in examining the joint probabilities $P(x_1,x_2,\dots,x_N|\hat{x}_1,\hat{x}_2,\dots,\hat{x}_N)$, where $\hat{x}_i$ can be either $\hat{u}_i$ or $\hat{v}_i$. Let us consider the following set of conditions:

\begin{equation}\label{eq:gcabello}
\begin{split}
&P_{A_{i}A_{i+1}}(v_r, 1|\hat{v}_i, \hat{u}_{i+1}) = 0,\\
&\forall r, v_r \neq d_i.    
\end{split}   
\end{equation}

We consider the subsystem $\mathcal{A}_{N+1}$ to be $\mathcal{A}_1$. Similar to the tripartite case, we define $p_N \equiv P(1, 1, \dots, 1|\hat{u}_1, \hat{u}_2, \dots, \hat{u}_N)$, $q_N \equiv P(d_1, d_2, \dots, d_N|\hat{v}_1, \hat{v}_2, \dots, \hat{v}_N)$ and their difference to be $\mathcal{S}_N:=p_N-q_N$. We will show that only non-local correlation can satisfy (\ref{eq:gcabello}) with $\mathcal{S}_N > 0$ and furthermore, this non-locality is genuine.

\begin{thm}\label{maintheorem}
If a non-local correlation satisfies the conditions in (\ref{eq:gcabello}) with $\mathcal{S}_N > 0$, then it is a genuine multipartite non-local correlation.
\end{thm}
\begin{proof}
We will first show that no local theory can satisfy conditions (\ref{eq:gcabello}) with $\mathcal{S}_N > 0$. Let us define a local hidden variable $\lambda$ that captures the full description of the system. This $\lambda$ takes values from a set $\Omega$. Moreover, we assert that the comprehensive depiction of the joint system is denoted by $\nu(\lambda)$. In a local description, there exist conditional probabilities $p_{A_{j}}(u_j | \hat{u}_j, \lambda)$ and $p_{A_{j}}(v_j | \hat{v}_j, \lambda)$ such that the joint probability can be expressed as follows: $P(x_1, x_2, \dots, x_N | \hat{x}_1, \hat{x}_2, \dots, \hat{x}_N) = \sum_{\lambda \in \Omega}\nu(\lambda) \prod_{j=1}^{N} p_{A_{j}}(x_j | \hat{x}_j, \lambda)$, where $\hat{x}_j \in \{\hat{u}_j, \hat{v}_j\}$. Now, let's assume that $\mathcal{S}_N > 0$, indicating that $p_N > 0$. This implies the existence of a nonempty subset $\Lambda \subseteq \Omega$, such that for all $\lambda \in \Lambda$, $P_{A_{i}}(1|\hat{u}_i,\lambda)=1$ for all $i$. Consequently, for all $r$ where $v_r \neq d_i$, we have $P_{A_{i}}(v_r|\hat{v}_i, \lambda) = 0$, which further implies $P(d_i|\hat{v}_i, \lambda) = 1$. Therefore, we can deduce that $P(d_1, d_2, \dots, d_N|\hat{v}_1, \hat{v}_2, \dots, \hat{v}_N, \Lambda) = P(1, 1, \dots, 1|\hat{u}_1, \hat{u}_2, \dots, \hat{u}_N) = p_N$. However, since $\Lambda \subseteq \Omega$, it follows that $P(d_1, d_2, \dots, d_N|\hat{v}_1, \hat{v}_2, \dots, \hat{v}_N, \Lambda) \leq P(d_1, d_2, \dots, d_N|\hat{v}_1, \hat{v}_2, \dots, \hat{v}_N)$. In other words, $\mathcal{S}_N \leq 0$, which contradicts the initial assumption of $\mathcal{S}_N$. Therefore, we can conclude that if any correlation satisfies (\ref{eq:gcabello}) and gives $\mathcal{S}_N > 0$, then it must be non-local.

 Now we will show the genuineness of the non-locality. Consider a $N$-partite probability distribution $P(x_1, x_2, \dots, x_N | \hat{x}_1, \hat{x}_2, \dots, \hat{x}_N)$ satisfying conditions in (\ref{eq:gcabello}) with $\mathcal{S}_N >0$, which is not genuine. Hence, it is a convex mixture of bi-local correlation. Let's consider a partition $(1, 2, \dots, m)$ versus $(m + 1, m + 2, \dots, N)$. Consequently, there must be at least one element within the convex combination that contributes non-zero value to the condition $P(1, 1, \dots, 1|\hat{u}_1, \hat{u}_2, \dots, \hat{u}_N) > 0$, given $p_N > q_N \geq 0$. Let's assume such an element takes the form {\small $P_1(x_1, \dots, x_m|\hat{x}_1, \dots,\hat{x}_m,\lambda)P_2(x_{m+1}, \dots, x_N|\hat{x}_{m+1}, \dots, \hat{x}_N,\lambda)$}. As $p_N > 0$ we must have, 

{\small\begin{equation}\label{eq:1}
\begin{split}
P_1(1,\dots,1|\hat{u}_1,\dots, \hat{u}_m,\lambda)P_2(1, \dots, 1|\hat{u}_{m+1},\dots,\hat{u}_N,\lambda) > 0
\end{split}  
\end{equation}}

From the equation (\ref{eq:gcabello}), we have $P_{A_{m}A_{m+1}}(v_m \neq d_m, 1|\hat{v}_m, \hat{u}_{m+1}) = 0$ and $P_{A_{N}A_{1}}(v_N \neq d_N, 1|\hat{v}_N, \hat{u}_{1}) = 0$. Hence, we have two important relations,

{\small \begin{equation}\label{eq:2}
\begin{split}
&P_1(v_m \neq d_m|\hat{v}_m,\lambda)P_2(1|\hat{u}_{m+1},\lambda) = 0,\\
&P_2(v_N \neq d_N |\hat{v}_N,\lambda)P_1(1|\hat{u}_1,\lambda) = 0.
\end{split}  
\end{equation}}

The form of $P_1$: Since by (\ref{eq:1}), we have $P_2(1, \dots, 1|\hat{u}_{m+1},\dots,\hat{u}_N,\lambda) \neq 0$, one definitely have $P_2(1|\hat{u}_{m+1},\lambda) \neq 0$. We also get $P_1(v_m \neq d_m|\hat{v}_m,\lambda)P_2(1|\hat{u}_{m+1},\lambda) = 0$ from equation (\ref{eq:2}). Hence, we get $P_1(v_m \neq d_m|\hat{v}_m,\lambda) = 0$ and $P_1(d_m|\hat{v}_m,\lambda) = 1$. Now as the NS principle is assumed at the level of distribution $P_1(x_1 x_2 \dots x_m| \hat{x}_1 \hat{x}_2 \dots \hat{x}_m, \lambda)$, a deterministic marginal of $P_1(x_1 x_2 \dots x_m| \hat{x}_1 \hat{x}_2 \dots \hat{x}_m, \lambda)$ indicates that $P_1(x_1 x_2 \dots x_m| \hat{x}_1 \hat{x}_2 \dots \hat{x}_m, \lambda)$ factorizes to certain $P'_1(x_1,\dots,x_{m-1}|\hat{x}_1,\dots,\hat{x}_{m-1}, \lambda)P^m_1(x_m|\hat{x}_m, \lambda)$.

From the equation (\ref{eq:gcabello}), we also have $P_{A_{m-1}A_{m}}(v_{m-1} \neq d_{m-1}, 1|\hat{v}_{m-1}, \hat{u}_{m}) = 0$. As $P_1(x_1 x_2 \dots x_m| \hat{x}_1 \hat{x}_2 \dots \hat{x}_m, \lambda)$ factorizes we have 
$P'_1(v_{m-1} \neq d_{m-1}|\hat{v}_{m-1}, \lambda)P^m_1(1|\hat{u}_m, \lambda) = 0$,
however again using equation (\ref{eq:1}) we have $P^m_1(1|\hat{u}_m, \lambda) \neq 0$ and it implies that $P'_1(d_{m-1}|\hat{v}_{m-1}, \lambda) = 1$. Using the same argument as above $P'_1$ factorizes. Hence eventually it is easy to see that $P_1$ fully factorizes.

The form of $P_2$: We have $P_1(1|\hat{u}_1, \lambda) \neq 0$ and $P_2(v_N \neq d_N|\hat{v}_N, \lambda)P_1(1|\hat{u}_1, \lambda) = 0$ from equation (\ref{eq:1}) and (\ref{eq:2}) respectively. Hence we have  $P_2(v_N \neq d_N|\hat{v}_N, \lambda) = 0$, thus $P_2(d_N|\hat{v}_N, \lambda) = 1$. Using the same argument for factorization of $P_1$, $P_2$ also fully factorizes.

Thus the distribution {\small $P_1(x_1, \dots, x_m|\hat{x}_1, \dots,\hat{x}_m,\lambda)P_2(x_{m+1}, \dots, x_N|\hat{x}_{m+1}, \dots, \hat{x}_N,\lambda)$} fully factorizes. Since the proof is the same for all cuts, the distribution $P(x_1, x_2, \dots, x_N | \hat{x}_1, \hat{x}_2, \dots, \hat{x}_N)$ becomes fully local. Such probability distribution cannot satisfy the equation (\ref{eq:gcabello}) for $\mathcal{S}_N > 0$. Hence we are done.
\end{proof}

\paragraph{Self-testing of genuine tripartite non-local states.---}\label{sec3}

We will now demonstrate self-testing for the tripartite case. Alice, Bob, and Charlie each choose measurements $\hat{x}_1 \in \{A_0, A_1\}$, $\hat{x}_2 \in \{B_0, B_1\}$, and $\hat{x}_3 \in \{C_0, C_1\}$, respectively. Their binary outcomes are denoted as $x_1 \in \{\pm 1\}$, $x_2 \in \{\pm 1\}$, and $x_3 \in \{\pm 1\}$. We define a tripartite correlation $P(x_1, x_2, x_3 | \hat{x}_1, \hat{x}_2, \hat{x}_3)$ that satisfies the equality constraint \ref{eq:cab3}. To illustrate genuine non-locality, consider that Alice, Bob, and Charlie share the following general pure three-qubit state,

\begin{equation}\label{eq:7}
\begin{split}
&|\Psi\rangle_{ABC}\\ &= a_{000}|000\rangle + a_{001}|001\rangle + a_{010}|010\rangle + a_{100}|100\rangle \\ & + a_{011}|011\rangle + a_{101}|101\rangle + a_{110}|110\rangle + a_{111}|111\rangle,
\end{split}  
\end{equation} 

such that $\sum_{i,j,k\in\{0,1\}}|a_{ijk}|^2=1$. We consider $\hat{x}_1 = |u_{x_1}^{+}\rangle\langle u_{x_1}^{+}| - |u_{x_1}^{-}\rangle\langle u_{x_1}^{-}|$, $\hat{x}_2 = |v_{x_2}^{+}\rangle\langle v_{x_2}^{+}| - |v_{x_2}^{-}\rangle\langle v_{x_2}^{-}|$ and $\hat{x}_3 = |w_{x_3}^{+}\rangle\langle w_{x_3}^{+}| - |w_{x_3}^{-}\rangle\langle w_{x_3}^{-}|$ to be the measurement of Alice, Bob and Charlie respectively. We can choose the following measurements with $\hat{x}_1\in\{A_0, A_1\}$ and $\hat{x}_2\in\{B_0, B_1\}$ and $\hat{x}_3\in\{C_0, C_1\}$ as,

{\small \begin{equation}\label{measure}
\begin{split}
A_0 &\equiv \left\{
\begin{array}{l}
|u_{A_0}^+ \rangle = |0\rangle,\\
|u_{A_0}^- \rangle = |1\rangle
\end{array}\right\},
\quad \\
A_1 &\equiv \left\{
\begin{array}{l}
|u_{A_1}^+ \rangle = \cos(\frac{\alpha_1}{2}) |0\rangle + e^{i\phi_1}\sin(\frac{\alpha_1}{2}) |1\rangle,\\
|u_{A_1}^- \rangle = -\sin(\frac{\alpha_1}{2}) |0\rangle + e^{i\phi_1}\cos(\frac{\alpha_1}{2}) |1\rangle
\end{array}\right\},
\quad \\
B_0 &\equiv \left\{
\begin{array}{l}
|v_{B_0}^+ \rangle = |0\rangle,\\
|v_{B_0}^- \rangle = |1\rangle
\end{array}\right\},
\quad \\
B_1 &\equiv \left\{
\begin{array}{l}
|v_{B_1}^+ \rangle = \cos(\frac{\alpha_2}{2}) |0\rangle + e^{i\phi_2}\sin(\frac{\alpha_2}{2}) |1\rangle,\\
|v_{B_1}^- \rangle = -\sin(\frac{\alpha_2}{2}) |0\rangle + e^{i\phi_2}\cos(\frac{\alpha_2}{2}) |1\rangle 
\end{array}\right\},
\quad \\
C_0 &\equiv \left\{
\begin{array}{l}
|w_{C_0}^+ \rangle = |0\rangle,\\
|w_{C_0}^- \rangle = |1\rangle
\end{array}\right\},
\quad \\
C_1 &\equiv \left\{
\begin{array}{l}
|w_{C_1}^+ \rangle = \cos(\frac{\alpha_3}{2}) |0\rangle + e^{i\phi_3}\sin(\frac{\alpha_3}{2}) |1\rangle,\\
|w_{C_1}^- \rangle = -\sin(\frac{\alpha_3}{2}) |0\rangle + e^{i\phi_3}\cos(\frac{\alpha_3}{2}) |1\rangle
\end{array}\right\}    
\end{split}    
\end{equation}}

where $0 < \alpha_1, \alpha_2, \alpha_3 < \pi$ and $0 \leq \phi_1, \phi_2, \phi_3 < 2\pi$. Without loss of generality we can choose observable $A_0 = B_0 = C_0= \sigma_z \equiv |0\rangle\langle 0| - |1\rangle\langle 1|$ as long as we consider the general form of the state $|\Psi\rangle_{ABC}$ and of the three observables $A_1$, $B_1$, $C_1$.

The class of pure three-qubit state that satisfies the condition (\ref{eq:cab3}) is of the form:

{\small \begin{equation}
\begin{split}
&|\Psi\rangle_{ABC}\\
& = a|000\rangle - a(e^{i\phi_3}\cot\frac{\alpha_3}{2} |001\rangle + e^{i\phi_2}\cot\frac{\alpha_2}{2} |010\rangle \\
& + e^{i\phi_1}\cot\frac{\alpha_1}{2}|100\rangle) + a(e^{i(\phi_2+\phi_3)}\cot\frac{\alpha_2}{2}\cot\frac{\alpha_3}{2}|011\rangle \\
& + e^{i(\phi_1+\phi_3)}\cot\frac{\alpha_1}{2}\cot\frac{\alpha_3}{2}|101\rangle + e^{i(\phi_1+\phi_2)}\cot\frac{\alpha_1}{2}\cot\frac{\alpha_2}{2}|110\rangle) \\ 
& + e^{i\delta} 
\sqrt{\begin{split}
&1-a^2(1+\cot^2\frac{\alpha_1}{2}+\cot^2\frac{\alpha_2}{2} \\
& +\cot^2\frac{\alpha_3}{2} +\cot^2\frac{\alpha_2}{2}\cot^2\frac{\alpha_3}{2} \\
& +\cot^2\frac{\alpha_1}{2}\cot^2\frac{\alpha_3}{2}+\cot^2\frac{\alpha_1}{2}\cot^2\frac{\alpha_2}{2})
\end{split}}|111\rangle
\end{split}
\end{equation}}
where $a$ and $\delta$ satisfy $0 \leq a \leq 1$ and $0 \leq \delta < 2\pi$. Then, maximizing $\mathcal{S}_3$ over these parameters yields, $\frac{1}{8} \left(2 \left(7 a \sqrt{1-7 a^2}-17 a^2\right)-1\right) \approx 0.02035$ for $\alpha_1 = \alpha_2 = \alpha_3 = \frac{\pi}{2}$, $a = \frac{1}{2} \sqrt{\frac{316-17 \sqrt{158}}{1106}} \approx 0.15207$ and $\delta = \phi_1+\phi_2+\phi_3$. The maximum value of $\mathcal{S}_3$ is attainable by the following pure three-qubit state,


\begin{equation*}
\begin{split}
&|\Psi_{max}^{\star 3}\rangle_{ABC}\\
&= a|000\rangle - a(e^{i\phi_3}|001\rangle + e^{i\phi_2} |010\rangle + e^{i\phi_1}|100\rangle)\\
& + a(e^{i(\phi_2+\phi_3)}|011\rangle  + e^{i(\phi_1+\phi_3)}|101\rangle + e^{i(\phi_1+\phi_2)}|110\rangle) \\ 
& + e^{i\delta} 
\sqrt{1-7a^2}|111\rangle
\end{split}
\end{equation*}

We will first show that the maximum achievable value of $\mathcal{S}_3$ over $\mathbb{C}^n \otimes \mathbb{C}^n \otimes \mathbb{C}^n$ is same as over $\mathbb{C}^2 \otimes \mathbb{C}^2 \otimes \mathbb{C}^2$. Further, we will see that the quantum state $|\Psi_{max}^{\star 3}\rangle_{ABC}$ can be self-tested.

\begin{thm}\label{optimal3} 
 The maximum achievable value of $\mathcal{S}_3$ among all three-qubit states represents the optimal value attainable within tri-partite quantum states of any finite dimension.
\end{thm}

\begin{proof}
Here we only present the outline of the proof as the details are quite similar to the proof given in \cite{RZS12}.

Consider a general tripartite state $\rho$ shared among Alice, Bob, and Charlie. The operator $\Pi_{x|\hat{x}}$ represents the outcome $x$ obtained by Alice when she measures the observable $\hat{x}$. Similarly, we define $\Pi_{y|\hat{y}}$ and $\Pi_{z|\hat{z}}$ for Bob and Charlie respectively. Thus, 

\begin{equation}
\begin{split}
P(x,y,z|\hat{x},\hat{y},\hat{z}) = \operatorname{Tr}[\rho (\Pi_{x|\hat{x}} \otimes \Pi_{y|\hat{y}} \otimes \Pi_{z|\hat{z}})].    
\end{split}
\end{equation}

Considering no constraints on dimensionality, we employ Neumark's dilation theorem, restricting our focus to projective measurements. Thus, the observables associated with Alice, Bob, and Charlie consist of Hermitian operators having eigenvalues of $\pm 1$, depicted as:

\begin{equation}
\begin{split}
\hat{x} &= (+1) \Pi_{+|\hat{x}} + (-1) \Pi_{-|\hat{x}}, \quad \hat{x} \in \{A_0, A_1\} \\
\hat{y} &= (+1) \Pi_{+|\hat{y}} + (-1) \Pi_{-|\hat{y}}, \quad \hat{y} \in \{B_0, B_1\}\\
\hat{z} &= (+1) \Pi_{+|\hat{z}} + (-1) \Pi_{-|\hat{z}}, \quad \hat{z} \in \{C_0, C_1\}   
\end{split}
\end{equation}

Now using the lemma stated in \cite{Mas06} applies to Alice's observable $\{A_0, A_1\}$ inducing a decomposition $H_A = \bigoplus_i H_{A}^i$, Bob's observable $\{B_0, B_1\}$ inducing a decomposition $H_B = \bigoplus_j H_{B}^j$, and Charlie's observable $\{C_0, C_1\}$ inducing a decomposition $H_C = \bigoplus_k H_{C}^k$. Consequently, the following expression is derived:

{\small \begin{equation}
\begin{split}
P(x,y,z|\hat{x},\hat{y},\hat{z}) &= \sum_{i,j,k} \lambda_{ijk} \operatorname{Tr}[\rho_{ijk} (\Pi_{x|\hat{x}} \otimes \Pi_{y|\hat{y}} \otimes \Pi_{z|\hat{z}})] \\
&\equiv \sum_{i,j,k} \lambda_{ijk} P_{ijk}(x,y,z|\hat{x},\hat{y},\hat{z})
\end{split}
\end{equation}}

where $\rho_{ijk} = \frac{\Pi_i \otimes \Pi_j \otimes \Pi_k \rho \Pi_i \otimes \Pi_j \otimes \Pi_k}{\lambda_{ijk}}$ is, at most, a three-qubit state. The coefficients $\lambda_{ijk}$ are determined as $\lambda_{ijk} = \operatorname{Tr}[\rho (\Pi_i \otimes \Pi_j \otimes \Pi_k)]$, with $\lambda_{ijk} \geq 0$ $\forall i,j,k$ and $\sum_{i,j,k} \lambda_{ijk} = 1$.

If the joint probability $P(x,y,z|\hat{x},\hat{y},\hat{z})$ meets the conditions outlined in (\ref{eq:cab3}), then it implies that the joint probability $P_{ijk}(x,y,z|\hat{x},\hat{y},\hat{z})$ will also adhere to the constraint equations. Consequently, the maximum attainable value of $\mathcal{S}_3$ across all three-qubit states denotes the optimal value achievable within tri-partite quantum states.
 
\end{proof}

\begin{thm}\label{self3}
If the maximum value of $\mathcal{S}_3$ is observed, then the state of the system is equivalent up to local isometries
to $|\sigma\rangle_{ABC} \otimes |\Psi_{max}^{\star 3}\rangle_{A'B'C'}$, where $|\Psi_{max}^{\star 3}\rangle$ attains the maximum value of $\mathcal{S}_3$  and $|\sigma\rangle$ is an arbitrary tripartite state.
\end{thm}

\begin{proof}
 We can choose eigenstates of the observables $A_0$, $B_0$, and $C_0$ to be computational eigenbasis:

\begin{equation*}
\begin{split}
\Pi_{+|A_0}^i &= \vert 2i \rangle \langle 2i \vert, \quad \Pi_{-|A_0}^i = \vert 2i+1 \rangle \langle 2i+1 \vert, \\
\Pi_{+|B_0}^j  &= \vert 2j \rangle \langle 2j \vert, \quad \Pi_{-|B_0}^j  = \vert 2j+1 \rangle \langle 2j+1 \vert,\\
\Pi_{+|C_0}^k &= \vert 2k \rangle \langle 2k \vert, \quad \Pi_{-|C_0}^k = \vert 2k+1 \rangle \langle 2k+1 \vert
\end{split}
\end{equation*}

where $i$, $j$, and $k$ belong to the set $\{0, 1, 2, \dots\}$. Now, the difference $p_{ijk} - q_{ijk}$, where $p_{ijk} = P_{ijk}(+, +, +|A_0, B_0, C_0)$ and $q_{ijk} = P_{ijk}(-, -, -|A_1, B_1, C_1)$, in the subspace $H_A^i \otimes H_B^j \otimes H_C^k$ can attain the maximum value if and only if $\rho_{ijk} = |\Psi_{max}^{\star 3}\rangle_{ijk}\langle \Psi_{max}^{\star 3}|$, where $|\Psi_{max}^{\star 3}\rangle_{ijk}$ represents the three-qubit state. Therefore, the unknown state $|\chi\rangle$ can give the maximum value of $\mathcal{S}_3$, if and only if,
\begin{equation*}
|\chi\rangle = \bigoplus_{i,j,k}\sqrt{\lambda_{ijk}}|\Psi_{max}^{\star 3}\rangle_{ijk}.
\end{equation*}

Hence, if we choose the local isometries in the following way,

\begin{align*}
&\Phi_A = \Phi_B = \Phi_C = \Phi, \\
&\Phi |2m,0\rangle_{XX'} \rightarrow |2m,0\rangle_{XX'}, \\
&\Phi |2m+1,0\rangle_{XX'} \rightarrow |2m,1\rangle_{XX'}, 
\end{align*}
where $XX' \in \{AA', BB', CC'\}$, then we have,

\begin{equation*}
\begin{split}
 &(\Phi_A\otimes\Phi_B\otimes\Phi_C)|\chi\rangle_{ABC}|000\rangle_{A'B'C'} \\&= |\sigma\rangle_{ABC} \otimes |\Psi_{max}^{\star 3}\rangle_{A'B'C'},   
\end{split}
\end{equation*}
where $|\sigma\rangle_{ABC}$ some junk state.
\end{proof}

\paragraph{Self-testing of genuine $N$-partite non-local states.---}\label{sec4}

Now we will provide self-testing of genuine $N$-partite non-local correlations. We will consider binary outcomes $\{+,-\}$ for all observables $\hat{x}_{i}$. So our conditions become:

\begin{equation}\label{eq:3}
\begin{split}
&P_{A_{i}A_{i+1}}(+, +|\hat{v}_i, \hat{u}_{i+1}) = 0,\\
&\forall i \leq N.    
\end{split}   
\end{equation}

We consider $p_N \equiv P(+, +, \dots, +|\hat{u}_1, \hat{u}_2, \dots, \hat{u}_N)$ and $q_N \equiv P(-, -, \dots, -|\hat{v}_1, \hat{v}_2, \dots, \hat{v}_N)$. We have already seen that no local theory will satisfy $\mathcal{S}_N (:=p_N - q_N) > 0$. Let us consider the following general $N$-qubit pure state shared between $N$ parties $\mathcal{A}_1, \mathcal{A}_2, \dots,$ and $\mathcal{A}_N$:

{\small \begin{equation}\label{eq:4}
\begin{split}
&|\Psi\rangle_{A_1...A_N} \\
& = a_{0...0}|0...0\rangle \\
&+ \sum_{\substack{x_i=1 \\ j\neq i, x_j = 0}}a_{x_1,\dots,x_i,\dots,x_N}|x_1,\dots,x_i,\dots,x_N\rangle \\ &+ \sum_{\substack{x_{i_1}=1,x_{i_2}=1 \\ j\neq \{i_1,i_2\}, x_j = 0}}a_{x_1,\dots,x_{i_1},\dots,x_{i_2},\dots,x_N}|x_1,\dots,x_{i_1},\dots,x_{i_2}\dots,x_N\rangle \\ & + \dots \\
&+\sum_{\substack{x_i=0 \\ j\neq i, x_j = 1}}a_{x_1,\dots,x_i,\dots,x_N}|x_1,\dots,x_i,\dots,x_N\rangle \\
&+a_{1...1}|1...1\rangle,
\end{split}  
\end{equation}}

such that $\sum_{i_1,\dots i_N\in\{0,1\}}|a_{i_1,\dots i_N}|^2=1$, and let the projective measurement of $A_i$ be $\hat{x}_i^j = |A_i^{x^j,+}\rangle\langle A_i^{x^j,+}| - |A_i^{x^j,-}\rangle\langle A_i^{x^j,-}|$. One can choose the basis for the measurements $\hat{x}_i^j$ where $j\in\{0, 1\}$ as follows:

{\small \begin{equation}\label{eq:5}
\begin{split}
\hat{x}_i^0 &\equiv \left\{
\begin{array}{l}
|A_i^{x^0,+} \rangle = |0\rangle,\\
|A_i^{x^0,-} \rangle = |1\rangle,
\end{array}\right\}
\quad \\
\hat{x}_i^1 &\equiv \left\{
\begin{array}{l}
|A_i^{x^1,+} \rangle = \cos(\frac{\alpha_i}{2}) |0\rangle + e^{i\phi_i}\sin(\frac{\alpha_i}{2}) |1\rangle,\\
|A_i^{x^1,-} \rangle = -\sin(\frac{\alpha_i}{2}) |0\rangle + e^{i\phi_i}\cos(\frac{\alpha_i}{2}) |1\rangle,
\end{array}\right\}
\end{split}  
\end{equation}}

where $0 < \alpha_i < \pi$ and $0 \leq \phi_i < 2\pi$. Now the state that will satisfy the equation (\ref{eq:3}) has amplitude of the following form:

{\small \begin{equation}\label{eq:6}
\begin{split}
& a_{0...0} = a,\\
& \forall i, x_i=1, j\neq i, x_j = 0,\\
& a_{x_1,\dots,x_i,\dots,x_N} = -a \cdot e^{i\phi_i} \cdot \cot{\frac{\alpha_i}{2}}\\
& \forall i_1, i_2, x_{i_1}=1, x_{i_2}=1 , j\neq \{i_1,i_2\}, x_j = 0,\\
& a_{x_1,\dots,x_{i_1},\dots,x_{i_2},\dots,x_N} = a \cdot e^{i(\phi_{i_1}+\phi_{i_2})} \cdot \cot{\frac{\alpha_{i_1}}{2}}\cdot \cot{\frac{\alpha_{i_2}}{2}},\\
& \dots \dots, \\
& a_{1...1} = e^{i\delta} \cdot 
\sqrt{
\begin{split}
1-&a^2_{0...0}- \sum_{\substack{x_i=1 \\ j\neq i, x_j = 0}}a^2_{x_1,\dots,x_i,\dots,x_N}\\
&-\sum_{\substack{x_{i_1}=1,x_{i_2}=1 \\ j\neq \{i_1,i_2\}, x_j = 0}}a^2_{x_1,\dots,x_{i_1},\dots,x_{i_2},\dots,x_N}\\
&- \dots - \sum_{\substack{x_i=0 \\ j\neq i, x_j = 1}}a^2_{x_1,\dots,x_i,\dots,x_N}
\end{split}
}
\end{split}  
\end{equation}} 

where $a$ and $\delta$ satisfy $0 \leq a \leq 1$ and $0 \leq \delta < 2\pi$.

We will find the maximum value of $\mathcal{S}_N$ using equation (\ref{eq:6}) like we have already found the maximum value for $\mathcal{S}_3$ for triparty case. The result in \cite{Mas06} works for any $N$-party with two input and two output scenarios. Hence, our calculations for triparty case easily generalize for $N$-party also. Hence, we have self-testing of the state which gives the maximum value for $\mathcal{S}_N$.

\begin{thm}\label{selfN}
If the maximum value of $\mathcal{S}_N$ is observed, then the state of the system is equivalent up to local isometries
to $|\sigma_{N}\rangle_{A_{1}\ldots A_{N}} \otimes |\Psi_{max}^{\star N}\rangle_{A'_{1}\ldots A'_{N}}$, where $|\Psi_{max}^{\star N}\rangle$ attains the maximum value of $\mathcal{S}_N$ and $|\sigma_{N}\rangle$ is an arbitrary $N$-partite state.
\end{thm}

\begin{proof}
We omit the proof as it is similar to the proof of Theorem \ref{self3}. 
\end{proof}

\paragraph{Numerical experiments:  Non-ideal constraints and entanglement cost.---}\label{sec5}
In an ideal case for three party scenario, the paradox demands that marginal probabilities should be equal to zero. But in a real experiment, this demand is very difficult to ensure. Hence a modified paradox that accounts real experiment can be of the form,

\begin{equation}\label{eq:cabni}
\begin{split}
& P_{AB}(++ | A_1B_0)  \leq \epsilon, \\
& P_{BC}(++ | B_1C_0)  \leq \epsilon, \\
& P_{AC}(++ | A_0C_1)  \leq \epsilon. \\
\end{split}  
\end{equation}

where $\epsilon \geq 0$ is some noise parameter. Without loss of generality, we can always choose the same noise parameter for each of the marginals. The local bound for the modified paradox takes the form,
{\small \begin{equation*}
\mathcal{S}_3^{\epsilon} \equiv P(+, +, + | A_0, B_0, C_0) - P(-, -, - | A_1, B_1, C_1) \leq 3\epsilon.
\end{equation*}}

Now we will show that even with this modified paradox the maximum value of $\mathcal{S}_3^{\epsilon}$ still can be achieved by performing projective measurements on pure three-qubits states. The analytical technique being introduced is not suitable for this scenario, necessitating a numerical demonstration. To begin, we will ascertain the upper limit of $\mathcal{S}_3^{\epsilon}$ by employing the NPA method pioneered by Navascues, Pironio, and Acin. Next, we show that for a particular class of three-qubit state, there exist projective measurements that achieve this upper bound with a very high order of accuracy. Note that, we only need to consider the noise parameter to be $0 \leq \epsilon < 1/3$. We obtain the upper bound on $\mathcal{S}_3^{\epsilon}$ by maximizing it over $\mathcal{Q}_3$ (NPA Hierarchy of level three). Next, we examined a specific category of pure states involving three qubits

{\footnotesize \begin{equation}\label{eq:state}
\begin{split}
&|\psi\rangle_{ABC} = a_{000}|000\rangle + a_{001}\left(e^{-i\phi_2}|010\rangle+e^{-i\phi_1}|100\rangle+e^{-i\phi_3}|001\rangle\right)\\
&+a_{011}\left(e^{-i(\phi_2+\phi_3)}|011\rangle+e^{-i(\phi_1+\phi_3)}|101\rangle+e^{-i(\phi_1+\phi_2)}|110\rangle\right)\\
&+a_{111}e^{-i(\phi_1+\phi_2+\phi_3)}|111\rangle,  
\end{split}
\end{equation}}
 
We performed numerical maximization of $\mathcal{S}_3^{\epsilon}$ across all states of the form \ref{eq:state} and measurement parameters, subject to the constraints outlined in equation (\ref{eq:cabni}). This process yielded $\mathcal{S}_3^{\epsilon}$ values for various noise parameter settings, ranging from $0$ to $\frac{1}{3}$. Fig. \ref{fig:3H} illustrates the maximum $\mathcal{S}_3^{\epsilon}$ values under $\mathcal{Q}_3$, three-qubit quantum states, and local model. Notably, we observed convergence between the maximum $\mathcal{S}_3^{\epsilon}$ values under $\mathcal{Q}_3$ and three-qubit quantum states, with an accuracy of approximately $10^{-6}$ within the error margin $0 \leq \epsilon \leq 0.08$.

We have also calculated different entanglement monotones for these states against $GHZ$ and $W$ states in three and four parties, respectively. These results are depicted in Fig. \ref{table}. These states require very few resources to $GHZ$ and $W$ states.

\begin{figure*}[tbh!]
  \centering
  \subfigure{\includegraphics[width = 0.48\textwidth]{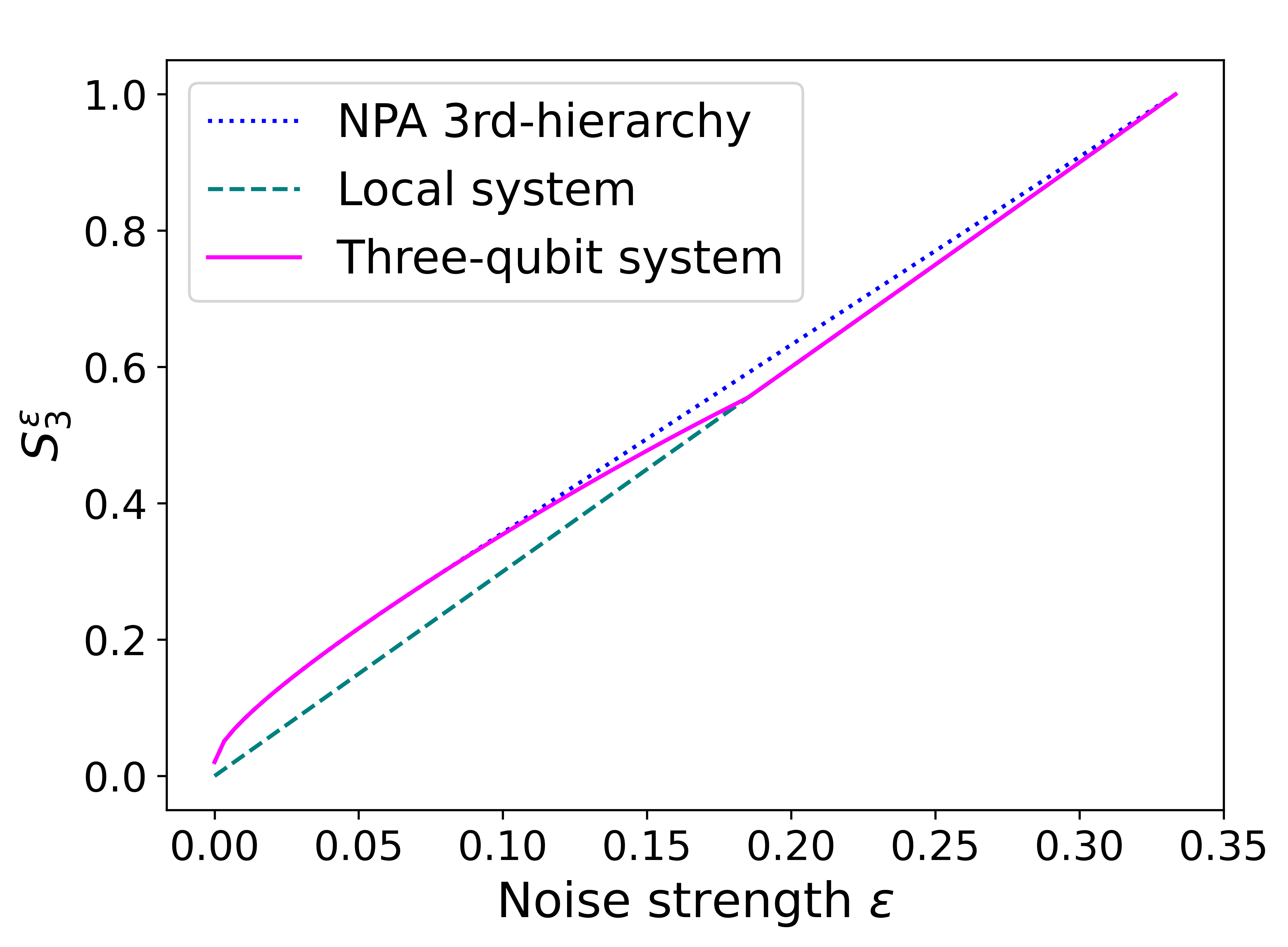}\label{fig:3H}}\quad
  \subfigure{\includegraphics[width = 0.465\textwidth]{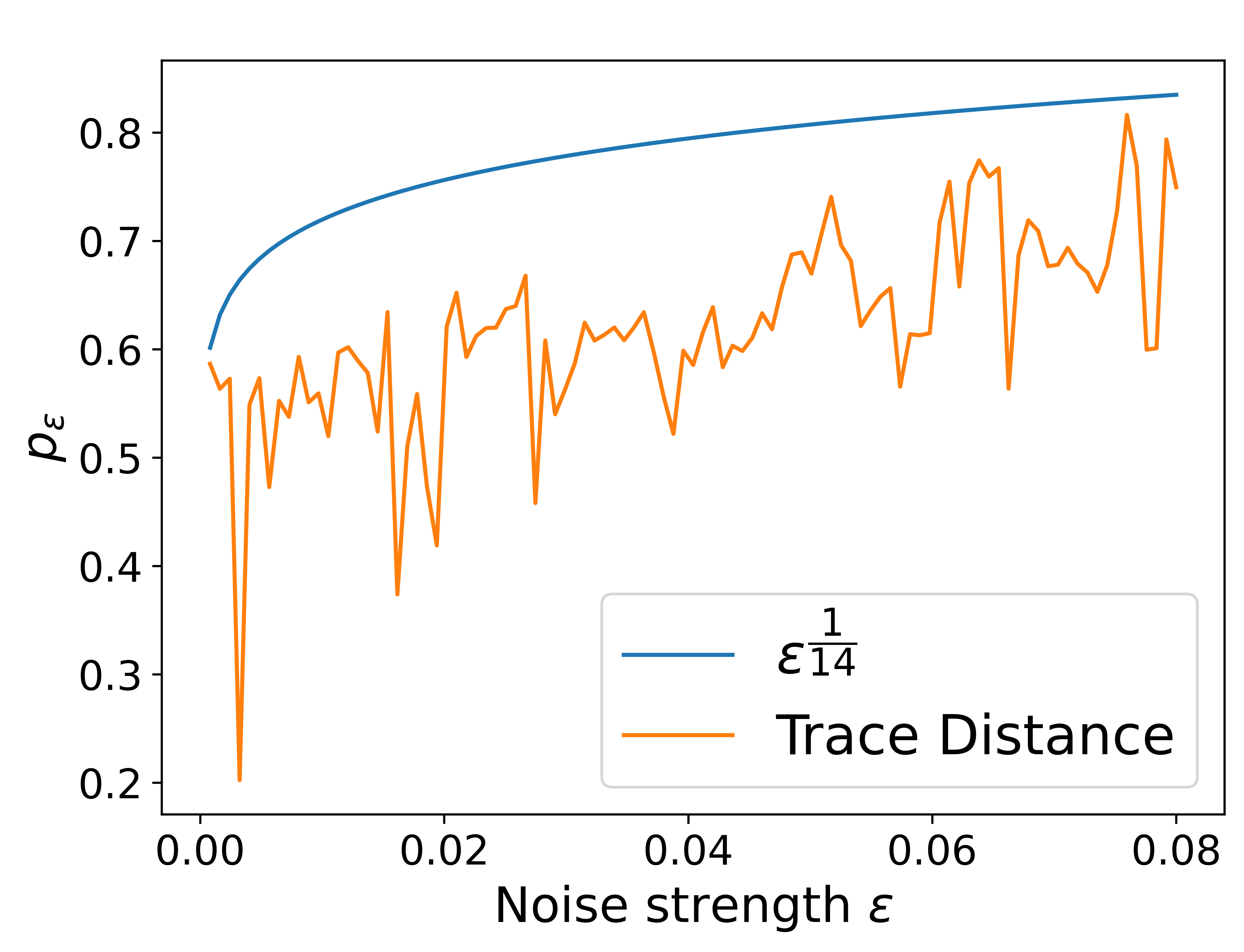}\label{fig:3G}}
  \caption{Robustness of our self-testing scheme.  Fig~\ref{fig:3H} depicts a plot of the maximum value of $\mathcal{S}_3^{\epsilon}$ of the paradox for triparty against the noise parameter $\epsilon$. The three-qubit system matches $\mathcal{Q}_3$ in the error bound range $0 \leq \epsilon \leq 0.08$. Fig~\ref{fig:3G} contains the variation of trace distance between ideal state and non-ideal state according to the range of $\epsilon$. Here $p_{\epsilon} = 0.5*\norm{|\Psi_{max}^\star\rangle_{ABC} - |\Psi^{\epsilon}\rangle_{ABC}}_1$, where $|\Psi^{\epsilon}\rangle_{ABC}$ varies with $\epsilon$. We see that the distance is upper bounded by $\epsilon^{\frac{1}{14}}.$}
    \label{fig1}
\end{figure*}

\begin{figure*}[tbh!]
  \centering
  \subfigure{\begin{tabular}{|c||c|c|c|}
\hline
 &GHZ&W&$|\Psi_{max}^{\star 3}\rangle$ \\ \hline
 Concurrence&  1&  0.94&  0.56\\ \hline
 Negativity&  0.5&  0.47&  0.28\\ \hline
 Log\ Negativity&  1&  0.96&  0.64\\ \hline
 Entanglement\ Entropy&  1&  0.92&  0.42\\ \hline
 R\'enyi\ Entanglement\ Entropy&  1&  0.96&  0.64\\ \hline
\end{tabular}\label{table:1}}\quad
  \subfigure{\begin{tabular}{|c||c|c|c|}
\hline
 &GHZ-4&W-4&$|\Psi_{max}^{\star 4}\rangle$ \\ \hline
 Concurrence&  1&  1&  0.39\\ \hline
 Negativity&  0.5&  0.43&  0.17\\ \hline
 Log\ Negativity&  1&  0.89&  0.43\\ \hline
 Entanglement\ Entropy&  1&  1&  0.24\\ \hline
 R\'enyi\ Entanglement\ Entropy&  1&  1&  0.48\\ \hline
\end{tabular}\label{table:2}}
  \caption{The quantum states that are self-tested via our scheme are non-maximally entangled.  Table~\ref{table:1} Different measure of entanglement monotone in 2 versus 1 bi-partition for three-qubit states, $GHZ$, $W$, and $|\Psi_{max}^{\star 3}\rangle$ respectively. $|\Psi_{max}^{\star 3}\rangle$ requires less resource than the other maximally entangled states. Table~\ref{table:2} Different measure of entanglement monotone in 2 versus 2 bi-partition for four-qubit states, $GHZ$, $W$, and $|\Psi_{max}^{\star 4}\rangle$ respectively. $|\Psi_{max}^{\star 4}\rangle$ requires less resource than the other maximally entangled states.}
    \label{table}
\end{figure*}

\paragraph{Discussion.---}\label{sec6}
To sum up, we've introduced a paradox similar to Hardy's paradox \cite{RWZ14}, designed for arbitrary $N$-partite systems where each parties have the choice of two measurements with arbitrary outcomes. Detecting a breach of our paradox witnesses genuine multipartite non-local behavior. Additionally, we've proved self-testing statements for genuinely non-local, but non-maximally multipartite entangled states.

Genuine multipartite non-local correlations are regarded as the most potent manifestation of genuine multipartite quantum correlations. The process of self-testing these quantum states serves to enhance the security and efficiency of quantum communication protocols and quantum key distribution. Moreover, it contributes to a deeper understanding of the intricate geometry of quantum boundaries. The inherent genuine non-locality displayed by our self-testing state makes it a valuable asset for secure quantum key distribution applications.

In a bipartite scenario, the Bell state is the unique maximally entangled state. However, in a multiparty setting, this uniqueness does not hold. Take, for instance, the tripartite scenario where two types of maximally entangled states exist: the $GHZ$ state and the $W$ state. Interestingly, when the party count exceeds three, there is no such well-defined characterization found in the existing literature. It is important to note that both the $GHZ$ and $W$ states qualify as genuine entangled states. Our numerical findings, showcased in Tables \ref{table:1} and \ref{table:2}, reveal that the quantum state leading to maximal violation in our inequality requires fewer resources compared to the $GHZ$ and $W$ states. It is worth emphasizing that in measuring entangled monotone, we employ bipartition, leveraging the presence of a unique maximally entangled state in the bipartite scenario.

In the context of robust self-testing, a typically expected outcome is that any quantum state corresponding to $\epsilon$ sub-optimal scenario should be approximately $\mathcal{O}(\sqrt{\epsilon})$ close to the quantum state associated with the ideal case upto local isometries in terms of trace distance. In this work, we have demonstrated partially robust self-testing, specifically showcasing that the trace distance between the quantum state linked to the ideal case and the one associated with a non-ideal scenario is upper-bounded by $\epsilon^{\frac{1}{14}}$. It's noteworthy that we consider the quantum state related to the non-ideal scenario as a three-qubit state. The numerical data that provides the substance for the aforementioned claim has been plotted in Figure \ref{fig:3G}. Moving forward, our immediate agenda involves developing a comprehensive robust self-testing scheme. Additionally, we plan to delve into the exploration of the amount of randomness that can be extracted from these correlations corresponding to various noise scales. This exploration will be the subject of our forthcoming research efforts.

\paragraph{Acknowledgement.---} Ranendu Adhikary acknowledges funding and support from Indian Statistical Institute, Kolkata. We would like to acknowledge stimulating discussions with Kishor Bharti, Ashutosh Rai, Tamal Guha, Subhendu B. Ghosh, and  Snehasish Roy Chowdhury.

\bibliography{cab}
\bibliographystyle{apsrev4-1}

\end{document}